\documentclass[10pt]{article}  
\usepackage{fullpage}

\usepackage{graphics} 
\usepackage{epsfig} 
\usepackage{mathptmx} 
\usepackage{times} 
\usepackage{amsmath,amsthm} 
\usepackage{amssymb}  
\usepackage{hyperref,url}
\usepackage{stmaryrd}

\theoremstyle{plain}
\newtheorem{theorem}{Theorem}[section]

\newtheorem{proposition}[theorem]{Proposition}
\newtheorem{corollary}[theorem]{Corollary}

\newtheorem{definition}[theorem]{Definition}

\theoremstyle{definition}
\newtheorem{example}[theorem]{Example}
\newtheorem{notation}[theorem]{Notation}

\newcommand{\flow}{\mathsf{flow}}
\newcommand{\jump}{\mathsf{jump}}
\newcommand{\inv}{\mathsf{inv}}
\newcommand{\init}{\mathsf{init}}

\newcommand{\reach}{\mathsf{Reach}}
\newcommand{\unsafe}{\mathsf{unsafe}}

\newcommand{\p}{\mathsf{P}}
\newcommand{\np}{\mathsf{NP}}
\newcommand{\enforce}{\mathsf{enforce}}

\newcommand{\lrf}{\mathcal{L}_{\mathbb{R}_{\mathcal{F}}}}
\newcommand{\dom}{\mathsf{dom}}
\newcommand{\der}{\mathrm{d}}

\title{\LARGE \bf Revisiting the Complexity of Stability of Continuous and Hybrid Systems}
\author{Sicun Gao \and Soonho Kong \and Edmund M. Clarke}

\begin{document}

\maketitle
\thispagestyle{empty}
\pagestyle{empty}

\begin{abstract}
We develop a general framework for obtaining upper bounds on the ``practical" computational complexity of stability problems, for a wide range of nonlinear continuous and hybrid systems. To do so, we describe stability properties of dynamical systems in first-order theories over the real numbers, and reduce stability problems to the $\delta$-decision problems of their descrptions. The framework allows us to give a precise characterization of the complexity of different notions of stability for nonlinear continuous and hybrid systems. We prove that bounded versions of the $\delta$-stability problems are generally decidable, and give upper bounds on their complexity. The unbounded versions are generally undecidable, for which we measure their degrees of unsolvability. 
\end{abstract}

\section{Introduction}

Stability of dynamical systems is a central topic in control theory. The computational nature of stability properties has been a topic of much recent investigation~\cite{DBLP:journals/corr/AhmadiP13,DBLP:journals/automatica/BlondelT99,DBLP:journals/automatica/BlondelT00,AAAthesis,DBLP:conf/hybrid/PrabhakarV13,DBLP:journals/corr/abs-1210-7420}. A focus of existing work is to establish various hardness results, i.e., lower bounds on complexity. It is shown that stability of simple systems is hard or impossible to solve algorithmically. Such results are proved by reducing combinatorial problems over graphs or matrices to stability problems, which can be analyzed with techniques of standard complexity theory. A limitation is that reduction techniques are usually not suitable for establishing upper bounds on complexity, and indeed most questions about upper bounds are open~\cite{AAAthesis}. 

Note that the existing approaches measure complexity of stability through {\em symbolic} manipulation of their descriptions. While doing so is suitable for establishing hardness results (for subclasses of the systems), it is at the core different from the practice in control theory, which is mostly based on {\em numerical} computations over real numbers. We argue that more general results for complexity of stability need to take into account of how real numbers and real functions are computed, as studied in computable analysis~\cite{CAbook,Kobook,16557}. However, while complexity for real functions is best measured in the model of Type 2 Turing machines~\cite{Kobook}, stability problems are still standard decision problems that should be measured in the standard complexity classes. Moreover, it is important to distinguish the difficulty with manipulating real numbers from the intrinsic complexity of control problems. For instance, if a real number {\em x} is represented numerically (as an infinite Cauchy sequence of rationals), determining whether "$x=0$" is already undecidable~\cite{CAbook}. Using such hardness in measuring complexity of practical control problems would be misleading, because in practice, the problems are always solved up to some nonzero error bound. That is, $|x|<\delta$ for a sufficiently small $\delta$ is what we need in practice, rather than the theoretically undecidable equality testing. The computational nature of the problem is very different with such a relaxation. 

We will show that $\delta$-decisions over the real numbers~\cite{DBLP:conf/lics/GaoAC12,DBLP:conf/cade/GaoAC12} provides a suitable framework for measuring the intrinsic complexity of control problems to address the issues discussed above. Within this framework, we can study the following version of stability problems. Given a dynamical system and an arbitrarily small $\delta\in \mathbb{Q}^+$, we ask for one of the following answers:
\begin{itemize}
\item The system is stable. 
\item The system is unstable under numerical perturbations bounded by $\delta$. 
\end{itemize}
We call this the $\delta$-stability problem. With this definition, we are able to give precise upper bounds for the ``practical complexity" of stability problems for a wide range of continuous and hybrid systems. We are able to prove results of the following type:
\begin{itemize}
\item Bounded Lyapunov $\delta$-stability resides in the complexity class $\mathsf{(\Pi^P_3)^C}$, where $\mathsf{C}$ is the complexity of continuous functions in the system. ($\mathsf{\Pi^P_3}$ denotes the complexity class in the polynomial hierarchy).
\item Bounded asymptotic $\delta$-stability resides in the complexity class $\mathsf{(\Sigma^P_4)^C}$. 
\item Unbounded Lyapunov $\delta$-stability is undecidable, whose degree of undecidability is $\mathsf{\Pi^0_1}$. Unbounded asymptotic $\delta$-stability is undecidable, whose degree of undecidability is in $\mathsf{\Sigma^0_2}$.
\item Lyapunov methods reduce problems into lower complexity classes such as $\mathsf{(\Sigma_2)^C}$.  
\end{itemize}
We believe these results are the first general characterization of the complexity of stability. Moreoever, the importance of the results is not just theoretical. The past decade has seen great advancement in decision procedures (SAT, QBF, and SMT solvers) that can handle many large instances of $\np$-hard problems. The complexity analysis shows the possibility of developing generic algorithmic approaches to control problems of nonlinear and hybrid systems.

In all, the main contributions of the paper are as follows:
\begin{itemize}
\item We define a framework for measuring the ``practical complexity" of stability problems for a wide range of nonlinear continuous and hybrid systems. To do so, we describe stability properties of systems as first-order formulas over the real numbers, and reduce stability problems to the $\delta$-decision problems of these formulas. 
\item The framework allows us to obtain a precise characterization of the complexity of different notions of stability that has not been discovered previously. We prove that bounded version of the stability problems are generally decidable, and give upper bounds on their complexity. The unbounded versions are generally undecidable, for which we measure their degrees of unsolvability. 
\end{itemize}
The paper is organized as follows. In Section II, we review definitions of complexity classes and some main results from computable analysis. In Section III, we review the theory of $\delta$-decisions over the reals and introduce the logic language that can encode a wide range of dynamical systems and properties. In Section IV, we study the complexity of stability of continuous systems. In Section V, we study the same questions for hybrid systems. We conclude in Section VI and suggest future directions. 

\section{Preliminaries}
\subsection{Oracle Machines, Polynomial and Arithmetic Hierarchies}

We review the basic definitions for complexity hierarchies. 

A {\em (set-) oracle Turing machine} $M$ extends an ordinary Turing machine with a special read/write tape called the {\em oracle tape}, and three special states $q_{\mathit{query}}$, $q_{\mathit{yes}}$, $q_{\mathit{no}}$. To execute $M$, we specify an oracle language $O\subseteq \{0,1\}^*$ in addition to the input $x$. Whenever $M$ enters the state $q_{\mathit{query}}$, it queries the oracle $O$ with the string $s$ on the oracle tape. If $s\in O$, then $M$ enters the state $q_{\mathit{yes}}$, otherwise it enters $q_{\mathit{no}}$. Regardless of the choice of $O$, a membership query to $O$ counts only as a single computation step. A {\em function-oracle Turing machine} is defined similarly except that when the machine enters the query state the oracle (given by a function $f:\{0,1\}^*\rightarrow\{0,1\}^*$) will erase the string $s$ on the query tape and write down $f(s)$. Note that such a machine must take $|f(s)|$ steps to read the output from the query tape. We write $M^O(x)$ (resp. $M^f(x)$) to denote the output of $M$ on input $x$ with oracle $O$ (resp. $f$). 

The polynomial hierarchy $\mathsf{PH}$ is a hierarchy of complexity classes that is defined through oracle computation. The base case are the well-known complexity classes $\mathsf{P}$ and $\np$. The classes in the hierarchy are recursively defined in the standard way:  
$$
\mathsf{\Sigma_0^P} = \mathsf{\Pi_0^P} = \mathsf{P}, \mathsf{\Sigma_{k+1}^P(A)} =\mathsf{NP^{\Sigma_k^P(A)}}, \mathsf{\Pi_{k+1}^P(A)}= \mathsf{coNP^{\Sigma_{k}^P(A)}}
$$
It is well-known that $\mathsf{PH}\subseteq \mathsf{PSPACE}$. If $\mathsf{P}\neq \mathsf{NP}$, then each class in the hierarchy contains harder problems than the previous ones. For undecidable problems, there exists an analogous arithmetic hierarchy. The base case is $\mathsf{\Sigma_1^0}$, which is the class of the halting problem. The other classes in the arithmetic hierarchy $\mathsf{\Pi_0^1}, \mathsf{\Sigma_2^0}, ...$ alternate in a similar way. The detailed definitions of polynomial and arithmetic hierarchy can be found in standard textbooks on recursion theory and computational complexity such as~\cite{arora09}. 

\subsection{Type 2 Computable Functions}

Given a finite alphabet $\Sigma$, let $\Sigma^*$ denote the set of finite strings and $\Sigma^{\omega}$ the set of infinite strings generated by $\Sigma$. For any $s_1, s_2\in \Sigma^*$, $\langle s_1,s_2\rangle$ denotes their concatenation. An integer $i\in \mathbb{Z}$ used as a string over $\{0,1\}$ has its conventional binary representation. The set of {\em dyadic rational numbers} is $\mathbb{D} = \{m/2^n: m\in \mathbb{Z}, n\in \mathbb{N}\}$. 

\paragraph{Computations over Infinite Strings} Standard computability theory studies operations over finite strings and does not consider real-valued functions. Real numbers can be encoded as infinite strings, and a theory of computability of real functions can be developed with oracle machines that perform operations using function-oracles encoding real numbers. This is the approach developed in Computable Analysis, a.k.a., Type 2 Computability. We will briefly review definitions and results of importance to us. Details can be found in the standard references~\cite{CAbook,Kobook,vasco}.

\begin{definition}[Names]
A name of $a\in \mathbb{R}$ is defined as a function $\mathcal{\gamma}_a: \mathbb{N}\rightarrow \mathbb{D}$ satisfying 
$$\forall i\in \mathbb{N}, |\gamma_a(i) - a|<2^{-i}.$$
For $\vec a\in \mathbb{R}^n$, $\gamma_{\vec a}(i) = \langle \gamma_{a_1}(i), ..., \gamma_{a_n}(i)\rangle$.  
\end{definition}
Thus the name of a real number is a sequence of dyadic rational numbers converging to it. For $\vec a\in \mathbb{R}^n$, we write $\Gamma(\vec a) = \{\gamma: \gamma\mbox{ is a name of }\vec a\}$. Noting that names are discrete functions, we can define
\begin{definition}[Computable Reals]
A real number $a\in \mathbb{R}$ is computable if it has a name $\gamma_{a}$ that is a computable function. 
\end{definition}

A real function $f$ is computable if there is a function-oracle Turing machine that can take any argument $x$ of $f$ as a function oracle, and output the value of $f(x)$ up to an arbitrary precision. 

\begin{definition}[Computable Functions]
We say a real function $f:\subseteq\mathbb{R}^n\rightarrow \mathbb{R}$ is Type 2 computable if there exists a function-oracle Turing machine $\mathcal{M}_f$, outputting dyadic rationals, such that for any $\vec x \in \dom(f)$, any name $\gamma_{\vec x}$ for $\vec x$, and any $i \in \mathbb{N}$, the output of $M_f^{\gamma_{\vec x}(i)}$ satisfies that  
$$|M_f^{\gamma_{\vec x}}(i) - f(\vec x)|<2^{-i},$$
which means that it approximates $f(\vec x)$ up to $2^{-i}$. 
\end{definition}

In the definition, $i$ specifies the desired error bound on the output of $M_f$ with respect to $f(\vec x)$. For any $\vec x\in \dom(f)$, $M_f$ has access to an oracle encoding the name $\gamma_{\vec x}$ of $\vec x$, and output a $2^{-i}$-approximation of $f(\vec x)$. In other words, the sequence 
$$M_f^{\gamma_{\vec x}}(1), M_f^{\gamma_{\vec x}}(2), ... $$
is a name of $f(\vec x)$. Intuitively, $f$ is computable if an arbitrarily good approximation of $f(\vec x)$ can be obtained using any good enough approximation to any $\vec x\in\dom(f)$.
\begin{proposition}[~\cite{CAbook}]
The following real functions are computable: addition, multiplication, absolute value, $\min$, $\max$, $\exp$, $\sin$ and solutions of Lipschitz-continuous ordinary differential equations. Compositions of computable functions are computable.  
\end{proposition}

A key property of the above notion of computability is that computable functions over reals must be continuous. In fact, over any compact set $D\subseteq \mathbb{R}^n$, computable functions are uniform continuous with a {\em computable modulus of continuity}. Intuitively, if a function has a computable uniform modulus of continuity, then fixing any desired error bound $2^{-i}$ on the output, we can compute a {\em global} precision $2^{-m_f(i)}$ on the inputs from $D$ such that using any $2^{-m_f(i)}$-approximation of any $\vec x\in D$, $f(\vec x)$ can be computed within the error bound. 

Complexity of real functions is usually defined over compact domains. Without loss of generality, we consider functions over $[0,1]$. Intuitively, a real function $f:[0,1]\rightarrow\mathbb{R}$ is (uniformly) $\mathsf{P}$-computable ($\mathsf{PSPACE}$-computable), if it is computable by an oracle Turing machine $M_{f}$ that halts in polynomial-time (polynomial-space) for every $i\in \mathbb{N}$ and every $\vec x\in \dom(f)$. The formal definition is as follows:
\begin{definition}[\cite{Kobook}]
A real function $f: [0,1]^n\rightarrow \mathbb{R}$ is in $\mathsf{P_{C[0,1]}}$ (resp. $\mathsf{PSPACE_{C[0,1]}}$) iff there exists a representation $(m_f, \theta_f)$ of $f$ such that
\begin{itemize}
\item $m_f$ is a polynomial function, and 
\item for any $d\in (\mathbb{D}\cap [0,1])^n$, $e\in \mathbb{D}$, and $i\in \mathbb{N}$, $\theta_f(d,i)$ is computable in time (resp. space) $O((\mathit{len}(d)+i)^k)$ for some constant $k$.
\end{itemize}
\end{definition}
\begin{proposition}
The following real functions all reside in Type 2 complexity class $\mathsf{P_{C[0,1]}}$: absolute value, polynomials, binary $\max$ and $\min$, $\exp$, $\sin$, and their bounded compositions. \end{proposition}
It is shown that solutions of Lipschitz-continuous differential equations are computable in $\mathsf{PSPACE_{C[0,1]}}$. In fact, it is shown to be $\mathsf{PSPACE}$-complete in the following sense. 
\begin{proposition}[\cite{Kawamura09}]
Let $g:[0,1]\times \mathbb{R}\rightarrow \mathbb{R}$ be polynomial-time computable and consider the initial value problem $\frac{df(t)}{dt} = g(t, f(t))$ for $f(0)=0$ and $t\in [0,1].$ Then computing the solution $f:[0,1]\rightarrow \mathbb{R}$ is in $\mathsf{PSPACE}$. Moreover, there exists $g$ such that computing f is $\mathsf{PSPACE}$-complete. 
\end{proposition}

\section{$\lrf$-Formulas and $\delta$-Decidability}

\subsection{$\lrf$-Formulas}
We will use a logical language over the real numbers that allows arbitrary {\em computable real functions}~\cite{CAbook}. We write $\lrf$ to represent this language. Intuitively, a real function is computable if it can be numerically simulated up to an arbitrary precision. For the purpose of this paper, it suffices to know that almost all the functions that are needed in describing hybrid systems are Type 2 computable, such as polynomials, exponentiation, logarithm, trigonometric functions, and solution functions of Lipschitz-continuous ordinary differential equations.

More formally, $\lrf = \langle \mathcal{F}, > \rangle$ represents the first-order signature over the reals with the set $\mathcal{F}$ of computable real functions, which contains all the functions mentioned above. Note that constants are included as 0-ary functions. $\lrf$-formulas are evaluated in the standard way over the structure $\mathbb{R}_{\mathcal{F}}= \langle \mathbb{R}, \mathcal{F}^{\mathbb{R}}, >^{\mathbb{R}}\rangle$. It is not hard to see that  we can put any $\lrf$-formula in a normal form, such that its atomic formulas are of the form $t(x_1,...,x_n)>0$ or $t(x_1,...,x_n)\geq 0$, with $t(x_1,...,x_n)$ composed of functions in $\mathcal{F}$. To avoid extra preprocessing of formulas, we can explicitly define $\mathcal{L}_{\mathcal{F}}$-formulas as follows.
\begin{definition}[$\lrf$-Formulas]
Let $\mathcal{F}$ be a collection of computable real functions. We define:
\begin{align*}
t& := x \; | \; f(t), \mbox{ where }f\in \mathcal{F} \mbox{ (constants are 0-ary functions)}\\
\varphi& := t> 0 \; | \; t\geq 0 \; | \; \varphi\wedge\varphi
\; | \; \varphi\vee\varphi \; | \; \exists x_i\varphi \; |\; \forall x_i\varphi.
\end{align*}
In this setting $\neg\varphi$ is regarded as an inductively defined operation
which replaces atomic formulas $t>0$ with $-t\geq 0$, atomic formulas $t\geq 0$
with $-t>0$, switches $\wedge$ and $\vee$, and switches $\forall$ and $\exists$.
\end{definition}
\begin{definition}[Bounded $\lrf$-Sentences]
We define the bounded quantifiers $\exists^{[u,v]}$ and $\forall^{[u,v]}$ as
\begin{eqnarray*}
\exists^{[u,v]}x.\varphi &=_{df}&\exists x. ( u \leq x \land x \leq v \wedge
\varphi)\\
\forall^{[u,v]}x.\varphi &=_{df}& \forall x. ( (u \leq x \land x \leq v)
\rightarrow \varphi)
\end{eqnarray*}
where $u$ and $v$ denote $\lrf$ terms, whose variables only
contain free variables in $\varphi$ excluding $x$. A {\em bounded $\lrf$-sentence} is
$Q_1^{[u_1,v_1]}x_1\cdots Q_n^{[u_n,v_n]}x_n\;\psi(x_1,...,x_n),$ where $Q_i^{[u_i,v_i]}$ are bounded quantifiers, and $\psi(x_1,...,x_n)$ is
quantifier-free.
\end{definition}
\subsection{$\delta$-Perturbations and $\delta$-Decidability}
\begin{definition}[$\delta$-Variants]\label{variants}
Let $\delta\in \mathbb{Q}^+\cup\{0\}$, and $\varphi$ an
$\lrf$-formula
$$\varphi: \ Q_1^{I_1}x_1\cdots Q_n^{I_n}x_n\;\psi[t_i(\vec x, \vec y)>0;
t_j(\vec x, \vec
y)\geq 0],$$ where $i\in\{1,...k\}$ and $j\in\{k+1,...,m\}$. The {\em
$\delta$-weakening} $\varphi^{\delta}$ of $\varphi$ is
defined as the result of replacing each atom $t_i > 0$ by $t_i >
-\delta$ and $t_j \geq 0$ by $t_j \geq -\delta$:
$$\varphi^{\delta}:\ Q_1^{I_1}x_1\cdots Q_n^{I_n}x_n\;\psi[t_i(\vec x, \vec
y)>-\delta; t_j(\vec x,
\vec y)\geq -\delta].$$
\end{definition}
It is easy to see that the perturbed formula is implied by the original formula. 
\begin{proposition}[(see \cite{DBLP:conf/lics/GaoAC12})]~\label{overap}
For any $\varphi$, we have $\varphi\rightarrow\varphi^{\delta}$.
\end{proposition}
In~\cite{DBLP:conf/lics/GaoAC12,DBLP:conf/cade/GaoAC12}, we have proved that the following $\delta$-decision problem is decidable, which is the basis of our framework.
\begin{theorem}[$\delta$-Decidability~\cite{DBLP:conf/lics/GaoAC12}]\label{delta-decide} Let $\delta\in\mathbb{Q}^+$ be
arbitrary. There is an algorithm which, given any bounded $\lrf$-sentence $\varphi$,
correctly returns one of the following two answers:
\begin{itemize}
\item $\delta$-$\mathsf{True}$: $\varphi^{\delta}$ is true.
\item $\mathsf{False}$: $\varphi$ is false.
\end{itemize}
When the two cases overlap, either answer is correct.
\end{theorem}
\begin{theorem}[Complexity~\cite{DBLP:conf/lics/GaoAC12}]\label{compmain}
Let $S$ be a class of $\lrf$-sentences, such that for any $\varphi$ in $S$, the terms in $\varphi$ are in Type 2 complexity class $\mathsf{C}$. Then, for any $\delta\in \mathbb{Q}^+$, the $\delta$-decision problem for bounded $\Sigma_n$-sentences in $S$ is in $\mathsf{(\Sigma_n^P)^C}$.
\end{theorem}

\section{Stability of Continuous Systems}

\subsection{$\lrf$-Representations}

Consider an $n$-dimensional autonomous ODE system
\begin{eqnarray}\label{ds}
\frac{\der x(t)}{\der t} = f(x(t))
\end{eqnarray}
where $f$ is Lipschitz-continuous and $x(0)\in \mathbb{R}^n$. We define the $\lrf$-representation of the system to be a logical formula that describes the all points on the trajectory of the dynamical system. 
\begin{definition}
We say the system (\ref{ds}) is $\lrf$-represented by an $\lrf$-formula $\flow(x_0, x_t, t)$, if for any $x(t)\in \mathbb{R}$, $x(t)$ is on the trajectory of the system iff the $\flow(x_0, x_t, t)$ is true. 
\end{definition}

From Picard-Lindel\"of iteration, we know that the $\lrf$-representation for continuous systems has an explicit form:

\begin{proposition}
The dynamical system in (\ref{ds}) has a trajectory that passes through $a\in \mathbb{R}$ iff the following $\lrf$-formula is true:
$$\flow(x_0, x_t, t)=_{df}\; (x_t = \int_0^t f(x(s))\der s + x_0)$$
\end{proposition}

\begin{proposition}
A continuous system has a $\lrf$-representation, when $f$ is a Type 2 computable function. 
\end{proposition}

Since $f$ can be any numerically computable function, this definition covers almost all dynamical systems of interest. We can now speak of the dynamical system (\ref{ds}) and its $\lrf$-representation $\flow(x_0, x_t, t)$ interchangeably. 

The $\delta$-perturbation on a system is defined through $\delta$-perturbations on its $\lrf$-representation. 
\begin{definition}
The $\delta$-perturbation of a system that is $\lrf$-represented by $\flow(x_0, x_t, t)$ is the system represented by $\flow^{\delta}(x_0, x_t, t)$. 
\end{definition}
To be clear, the $\flow$ formula has an explicit definition:
\begin{proposition}
The $\delta$-perturbation of the system (\ref{ds}) is represented by 
$$\flow^{\delta} =_{df}\; |x_t- (\int_0^t f(x(s))\der s + x_0)|<\delta.$$
\end{proposition}
Note that the $\delta$-perturbed system is always an overapproximation of the original system:
\begin{proposition}
We have $\llbracket\flow\rrbracket\subseteq \llbracket \flow^{\delta}\rrbracket$. 
\end{proposition}

\subsection{Complexity of Lyapunov Stability}

We first study stability in the sense of Lyapunov, which we can write stable i.s.L. Following standard definition, a system is stable i.s.L. if given any $\varepsilon$, there exists $\delta$ such that for any initial value $x_0$ that is within $\delta$ from the origin, the system stays in $\varepsilon$-distance from the origin. The $\lrf$-representation of stability in the sense of Lyapunov is naturally the following formula. 
\begin{definition}[{\sf L\_stable}]
We encode conditions for Lyapunov stability with the formula {\sf L\_stable} as follows. 
\begin{eqnarray*}
& &\forall^{[0,\infty)} \varepsilon\exists^{[0,\varepsilon]} \delta \forall^{[0,\infty)} t\forall x_0\forall x_t .\; (||x_0||<\delta \wedge x_t = \int_0^t f(s)ds + x_0 )\rightarrow ||x_t||<\varepsilon.
\end{eqnarray*}
The {\em bounded form} of {\sf L\_stable} is defined by bounding the quantifiers in the formula as follows:
\begin{eqnarray*}
& &\forall^{[0, e]} \varepsilon\exists^{[0,\varepsilon]} \delta \forall^{[0,T]} t\forall^X x_0\forall^X x_t. \;(||x_0||<\delta \wedge x_t = \int_0^t f(s)ds + x_0 )\rightarrow ||x_t||<\varepsilon, 
\end{eqnarray*}
where $e, T\in \mathbb{R}^+$ and $X$ is a compact set.
\end{definition}

It is not hard to see that the formula encodes the definition of stability in the sense of Lyapunov. 
\begin{proposition}
The origin is a stable equilibrium point iff {\sf L\_stable} is true. 
\end{proposition}

We can now define the $\delta$-stability problem using the $\lrf$-representation.  
\begin{definition}[$\delta$-Stability i.s.L.]\label{sl}
The $\delta$-stability problem i.s.L. asks for one of the following answers:
\begin{itemize}
\item {\sf stable}: The system is stable i.s.L. ({\sf L\_stable} is true). 
\item {\sf $\delta$-unstable}: Some $\delta$-perturbation of {\sf L\_stable} is false. 
\end{itemize}
We defined the {\em bounded} $\delta$-stability problem by replacing {\sf L\_stable} with the bounded form of {\sf L\_stable} in the definition. 
\end{definition}

Now, using the complexity of the formulas, we have the following complexity results for the bounded version of Lyapunov stability. 
\begin{theorem}[Complexity]
Suppose all terms in the $\lrf$-representation of a system are in Type 2 complexity class $\mathsf{C}$.  Then the bounded $\delta$-stability problem i.s.L. resides in complexity class $\mathsf{(\Pi^P_3)^C}$. 
\end{theorem}
\begin{proof}
The $\lrf$-formula {\sf L\_stable} is a $\sigma_3$ formula. By Definition~\ref{sl}, the $\delta$-stability problem is equivalent to the $\delta$-decision problem of the formula {\sf L\_stable}.  Following Theorem~\ref{compmain}, we have that the complexity of the $\delta$-decision problem for the bounded form of {\sf L\_stable} is in $\mathsf{(\Pi^P_3)^C}$. Consequently, the bounded $\delta$-stability problem i.s.L. resides in $\mathsf{(\Pi^P_3)^C}$. 
\end{proof}

Following the complexity for Lipschitz-continuous ODEs, we have an upper bound for the complexity of a wide range of systems. 
\begin{corollary}
Suppose that in the system (\ref{ds}), $f$ is a Type 2 polynomial-time computable function. Then the bounded $\delta$-stability problem i.s.L. is in $\mathsf{PSPACE}$.  
\end{corollary}
\begin{proof}
The $\lrf$-representation $\flow$ can be evaluated in $\mathsf{PSAPCE}$. Since $\mathsf{(\Pi^P_3)^{PSPACE}} \subseteq \mathsf{PSPACE}$, we know that the problem resides in $\mathsf{PSPACE}$. 
\end{proof}
We have mentioned that most of common functions and their compositions are polynomial-time computable: polynomials, trigonometric functions, exponential functions, etc. Consequently, for most nonlinear continuous systems of practical interest, the stability problem is in $\mathsf{PSPACE}$. 

The unbounded case involves testing the bounded formula for longer and longer time durations. Thus, it is still undecidable. We can obtain the degree of undecidability of the unbounded case from the logical encoding. 
\begin{theorem}
The unbounded Lyapunov $\delta$-stability problem is in $\Pi^0_1$. 
\end{theorem}
\begin{proof}
We compute $\delta$-decisions of the bounded form of the formula {\sf L\_stable} for increasingly larger time bound $T$. If for any $T$ the formula is $\delta$-false, then the system is $\delta$-unstable. On the other hand, we will not be able to confirm that the system is stable as $T$ approaches infinity. Thus, the problem is in $\Pi^0_1$ of the arithmetic hierarchy. \end{proof}

\subsection{Complexiy of Asymptotic Stability}

Following standard terminology, we say a system is asymptotically stable if it is Lyapunov stable, and there exists some bound on the perturbation in the initial state such that the system will converge to the origin eventually. We now study the complexity of this problem. 

First, since asymptotic stability involves properties of the system at the limit, we need to be express that as an $\lrf$-formula, as follows. 
\begin{definition}
We define the following formula for $\lim_{x\rightarrow \infty}(f(x), c)$
$$\lim_{x\rightarrow \infty}(f(x), c) =_{df}\; \forall^{[0,\infty)} \varepsilon \exists^{[0,\infty)} x \forall^{[x,\infty)}x'  \; (|f(x) - c|<\varepsilon).$$
We can use the conventional notation $\lim_{x\rightarrow \infty} f(x) = c$. Also, for convergence at a point $a\in \mathbb{R}^+$, we define 
$$\lim_{x\rightarrow a}(f(x), c) =_{df}\; \forall^{[0,\infty)} \varepsilon \exists^{[0,\infty)} \delta \forall^{[a-\delta,a+\delta]}x  \; (|f(x) - c|<\varepsilon).$$
Note that here the quantification on $\varepsilon$ and $\delta$ can be easily bounded, since we do not need to consider $\varepsilon$ and $\delta$ that are very large. Although further parameterization on the bounds are needed, for notational simplicity we simply treat this formula as a bounded $\lrf$-formula. 
\end{definition}

Now, asymptotic stability is defined as:
\begin{definition}[{\sf A\_stable}]
We define {\sf A\_stable} to be the following $\lrf$-formula
\begin{eqnarray*}
& &\forall^{[0,\infty)} \varepsilon\exists^{[0,\varepsilon]} \delta\forall^{[0,\infty)} t\forall x_0\forall x_t\Big((||x_0||<\delta \wedge x_t = \int_0^t f(s)ds + x_0 )\rightarrow ||x_t||<\varepsilon\Big)\\
& &\hspace{2cm}\wedge \exists^{[0,\infty)} \delta'  \forall^{[0,\infty)} t\forall x_0\forall x_t\Big( (||x_0||<\delta'\wedge x_t = \int_0^t f(s)ds + x_0 )\rightarrow \lim_{t\rightarrow \infty} ||x_t|| = 0\Big). 
\end{eqnarray*}
The bounded form of {\sf A\_stable} is defined as:
\begin{eqnarray*}
& &\forall^{[0,e]} \varepsilon\exists^{[0,\varepsilon]} \delta\forall^{[0,T]} t\forall^X x_0\forall^X x_t\Big((||x_0||<\delta \wedge x_t = \int_0^t f(s)ds + x_0 )\rightarrow ||x_t||<\varepsilon\Big)\\
& &\hspace{2cm}\wedge \exists^{[0,d]} \delta'  \forall^{[0,T']} t\forall^X x_0\forall^X x_t\Big( (||x_0||<\delta'\wedge x_t = \int_0^t f(s)ds + x_0 )\rightarrow \lim_{t\rightarrow T'} ||x_t|| = 0\Big) 
\end{eqnarray*}
where $e,T,T',d\in \mathbb{R}^+$ and $X$ is a compact set. 
\end{definition}
\begin{proposition}
The origin is asymptotically stable for a system iff the formula {\sf A\_stable} is true. 
\end{proposition}
We can now define the $\delta$-stability problem using the $\lrf$-representation.  
\begin{definition}[Asymptotic $\delta$-Stability]\label{sl}
The $\delta$-stability problem i.s.A. asks for one of the following answers:
\begin{itemize}
\item {\sf stable}: The system is stable i.s.A. ({\sf A\_stable} is true). 
\item {\sf $\delta$-unstable}: Some $\delta$-perturbation of {\sf A\_stable} is false. 
\end{itemize}
We defined the {\em bounded} $\delta$-stability problem by replacing {\sf A\_stable} with the bounded form of {\sf A\_stable} in the definition. 
\end{definition}

We can now obtain complexity results for the problem. 
\begin{theorem}
Suppose all terms in the $\lrf$-representation of a system are in Type 2 complexity class $\mathsf{C}$. Then bounded asymptotic $\delta$-stability is in $\mathsf{{(\Sigma_4^P)}^C}$.
\end{theorem}
\begin{proof}
The complexity of the formula is higher than the one encoding Lyapunov stability, because of the quantification structure in the encoding of the limit. After rearranging the formula, we have
\begin{eqnarray*}
& &\forall^{[0,e]} \varepsilon\exists^{[0,\varepsilon]} \delta\forall^{[0,T]} t\forall^X x_0\forall^X x_t\Big((||x_0||<\delta \wedge x_t = \int_0^t f(s)ds + x_0 )\rightarrow ||x_t||<\varepsilon\Big)\\
& &\hspace{.5cm}\wedge\; \exists^{[0,d]} \delta'  \forall^{[0,T']} t\forall^X x_0\forall^X x_t\forall^{[0,e']} \varepsilon' \exists^{[0,d']} \delta'' \forall^{[-\delta'',+\delta'']}t \Big( (||x_0||<\delta'\wedge x_t = \int_0^t f(s)ds + x_0 )\rightarrow ||x_t|| < \varepsilon'\Big) 
\end{eqnarray*}
This is a $\Sigma_4$-formula. Following Theorem~\ref{compmain} we know that the problem resides in $\mathsf{{(\Sigma_4^P)}^C}$. 
\end{proof}
The degree of undecidability for the unbounded version is, however, different from Lyapunov stability. This is because we need to find the bound of perturbation that ensures the convergence to the origin. 
\begin{corollary}
Unbounded asymptotic $\delta$-stability is in $\mathsf{\Sigma^0_2}$. 
\end{corollary}
\begin{proof}
In the formula {\sf A\_stable}, we need to incrementally search for a value for $\delta'$. Each of the value corresponds to an unbounded search for the time bound, which is similar to the case of Lyapunov complexity. Thus, we need to solve unbounded $\exists\forall$ quantification, which means the unbounded problem is in $\Sigma^0_2$ of the arithmetic hierarchy. 
\end{proof}
It is probably interesting to note that the problem $\p\neq \np$ has the same degree of undecidability. 

There is also the notion of ``asymptotic stability in the large," which ensures that for any perturbation on $x(0)$, the system will stabilize. The quantification turns out to be slightly different:
\begin{proposition}[Asymptotic Stability in the Large]
The origin is an asymptotically stable equilibrium point iff the following $\lrf$-formula is true
\begin{eqnarray*}
& &\forall^{[0,\infty)} \varepsilon\exists^{[0,\varepsilon]} \delta\forall^{[0,\infty)} t\forall x_0\forall x_t\Big((||x_0||<\delta \wedge x_t = \int_0^t f(s)ds + x_0 )\rightarrow ||x_t||<\varepsilon\Big)\\
& &\hspace{2cm}\wedge \forall^{[0,\infty)} \delta'  \forall^{[0,\infty)} t\forall x_0\forall x_t\Big( (||x_0||<\delta'\wedge x_t = \int_0^t f(s)ds + x_0 )\rightarrow \lim_{t\rightarrow \infty} ||x_t|| = 0\Big). 
\end{eqnarray*}
\end{proposition}
Computationally, this is in fact a simpler task than asymptotic stability. We state the following result without duplicating the proofs. 
\begin{theorem}
Suppose all terms in the $\lrf$-representation of a system are in Type 2 complexity class $\mathsf{C}$. Then bounded asymptotic $\delta$-stability in the large is in $\mathsf{(\Pi^P_3)^C}$. The unbounded case resides in $\mathsf{\Pi^0_1}$.  
\end{theorem}


\subsection{Complexity of Lyapunov Methods}

We show that Lyapunov methods reduce the complexity of stability problems. We only discuss the first-order encodings of the problems, in which a Lyapunov function is considered with a template function with unspecified parameters. 

\begin{proposition} Consider the dynamical system (\ref{ds}). Let $V(p,x)$ be a function, parameterized by $p$, whose partial derivative ${\partial V}/{\partial x}$ is a Type 2 computable function. Let $D$ be the parameter space for $p$ and $X$ be the state space of $x$. We then have
\begin{itemize}
\item The following $\lrf$-formula is a sufficient condition for stability in the sense of Lyapunov
$$\exists p^D\forall^X x\; \bigg(V(p,0)=0\wedge (x\neq 0\rightarrow V(p,x)>0)\wedge\frac{\partial V(p,x)}{\partial x}f(x)\leq 0\bigg)$$
\item The following is a sufficient condition for asymptotic stability:
\begin{eqnarray*}
\exists p^D\forall^X x\;\bigg(V(p,0)=0\wedge (x\neq 0\rightarrow V(p,x)>0)\wedge\Big(x=0\rightarrow\frac{\partial V(p,x)}{\partial x}f(x)= 0\bigg)
\wedge \bigg(x\neq 0 \rightarrow \frac{\partial V(p,x)}{\partial x}f(x)< 0\Big)\bigg)
\end{eqnarray*}
\end{itemize}
\end{proposition}

\begin{definition}[$\delta$-Complete Lyapunov Test]
Let $V(p,x)$ be a proposed template for Lyapunov function. The $\delta$-complete Lyapunov test asks for one of the following answers:
\begin{itemize}
\item {\sf Success}: There exists an assignment to $p$ such that the Lyapunov function witness stability of the system. 
\item {\sf $\delta$-Fail}: The Lyapunov conditions fail under $\delta$-perturbations for all possible parameterizations of $V(p,x)$ in the parameter space $D$.
\end{itemize}
\end{definition}
\begin{theorem}
Suppose all terms in the $\lrf$-representation of the Lyapunov conditions are in Type 2 complexity class $\mathsf{C}$. The complexity of bounded $\delta$-complete Lyapunov methods is in $\mathsf{(\Sigma^P_2)^C}$. 
\end{theorem}

It is clear that for the fully unbounded case (where both $D$ and $X$ are unbounded), undecidability comes from the search in larger and larger parameter and state space. 
\begin{corollary}
The unbounded $\delta$-complete Lyapunov test for an unbounded system is in $\mathsf{\Sigma_2^0}$. 
\end{corollary}

\section{Stability of Hybrid Systems}

An important benefit of using logic formulas for describing systems is that discrete changes can be naturally represented.  Although the discrete components significantly complicates the $\lrf$-representations of the problems, they do not change the quantification structure of the encodings. Thus, we will see that the complexity upper bound of the continuous systems mostly carry over to the case of hybrid systems as well. On the other hand, it is indeed easier to show hardness results (lower-bound) using logical operations, and in this sense hybrid systems are intrinsically more complicated than continuous systems. 

\subsection{{\large$\lrf$}-Representations of Hybrid Systems}\label{language}

We first show that $\lrf$-formulas can concisely represent hybrid automata.
\begin{definition}\label{lrf-definition}
A hybrid automaton in $\lrf$-representation is a tuple
\begin{multline*}
H = \langle X, Q, \{{\flow}_q(\vec x, \vec y, t): q\in Q\},\{\inv_q(\vec x): q\in Q\},
\{\jump_{q\rightarrow q'}(\vec x, \vec y): q,q'\in Q\},\{\init_q(\vec x): q\in Q\}\rangle
\end{multline*}
where $X\subseteq \mathbb{R}^n$ for some $n\in \mathbb{N}$, $Q=\{q_1,...,q_m\}$ is a finite set of modes, and the other components are finite sets of quantifier-free $\lrf$-formulas.
\end{definition}
\begin{notation}
For any hybrid system $H$, we write $X(H)$, $\flow(H)$, etc. to denote its corresponding components.
\end{notation}
Almost all hybrid systems studied in the existing literature can be defined by restricting the set of functions $\mathcal{F}$ in the signature. For instance,
\begin{example}[Linear and Polynomial Hybrid Automata] Let $\mathcal{F}^{\mathrm{lin}} = \{+\}\cup \mathbb{Q}$ and $\mathcal{F}^{\mathrm{poly}}=\{\times\}\cup\mathcal{F}^{\mathrm{lin}}$. Rational numbers are considered as 0-ary functions. In existing literature, $H$ is a {\em linear hybrid automaton} if it has an $\mathcal{L}_{\mathbb{R}_{\mathcal{F}^{\mathrm{lin}}}}$-representation, and a {\em polynomial hybrid automaton} if it has an $\mathcal{L}_{\mathbb{R}_{\mathcal{F}^{\mathrm{poly}}}}$-representation.
\end{example}
\begin{example}[Nonlinear Bouncing Ball]
The bouncing ball is a standard hybrid system model. Its nonlinear version (with air drag) can be $\lrf$-represented as follows:
\begin{itemize}
\item $X = \mathbb{R}^2$ and $Q = \{q_u, q_d\}$. We use $q_u$ to represent bounce-back mode and $q_d$ the falling mode.
\item $\flow = \{\flow_{q_u}(x_0, v_0, x_t, v_t, t), \flow_{q_d}(x_0, v_0, x_t, v_t, t)\}$. We use $x$ to denote the height of the ball and $v$ its velocity. Instead of using time derivatives, we can directly write the flows as integrals over time, using $\lrf$-formulas:
\begin{itemize}
\item $\flow_{q_u}(x_0, v_0, x_t, v_t, t)$ defines the dynamics in the bounce-back phase:
$$(x_t = x_0 + \int_0^{t} v(s) ds) \wedge (v_t = v_0 + \int_0^t g(1-\beta v(s)^2) ds)$$
\item $\flow_{q_d}(x_0, v_0, x_t, v_t, t)$ defines the dynamics in the falling phase:
$$(x_t = x_0 + \int_0^{t} v(s) ds) \wedge (v_t = v_0 + \int_0^t g(1+\beta v(s)^2) ds)$$
\end{itemize}where
$\beta$ is a constant. Again, note that the integration terms define Type 2 computable functions.
\item $\jump = \{\jump_{q_u \rightarrow q_d} (x, v, x', v'), \jump_{q_d \rightarrow q_u} (x, v, x', v')\}$ where
\begin{itemize}
 \item $\jump_{q_u \rightarrow q_d} (x, v, x', v')$ is $(v= 0 \wedge x' = x \wedge v' = v)$.
\item $\jump_{q_d \rightarrow q_u} (x, v, x', v')$ is $(x= 0 \wedge v' = \alpha v\wedge x'=x)$,  for some constant $\alpha$.
\end{itemize}
\vspace{0.1cm}
\item $\init_{q_d}: (x=10 \wedge v=0)$ and $\init_{q_u}: \bot$.
\item $\inv_{q_d}: (x>=0 \wedge v>=0)$ and $\inv_{q_u}: (x>=0 \wedge v<=0)$.
\end{itemize}
\end{example}

Trajectories of hybrid systems combine continuous flows and discrete jumps. This motivates the use of a hybrid time domain, with which we can keep track of both the discrete changes and the duration of each continuous flow. A hybrid time domain is a sequence of closed intervals on the real line, and a hybrid trajectory is a mapping from the time domain to the Euclidean space. 

We now define $\delta$-perturbations on hybrid automata directly through perturbations on the logic formulas in their $\lrf$-representations. For any set $S$ of $\lrf$-formulas, we write $S^{\delta}$ to denote the set containing the $\delta$-perturbations of all elements of $S$.
\begin{definition}[$\delta$-Weakening of Hybrid Automata] Let $\delta\in\mathbb{Q}^+\cup\{0\}$ be arbitrary. Suppose
$$H = \langle X, Q, \flow, \jump, \inv, \init\rangle$$
is an $\lrf$-representation of hybrid system $H$. The {\em $\delta$-weakening} of $H$ is
$$H^{\delta} = \langle X, Q, \flow^{\delta}, \jump^{\delta}, \inv^{\delta}, \init^{\delta}\rangle$$
which is obtained by weakening all formulas in the $\lrf$-representations of $H$.
\end{definition}
\begin{example}
The $\delta$-weakening of the bouncing ball automaton is obtained by weakening the formulas in its description. For instance, $\flow_{q_u}^{\delta}(x_0, v_0, x_t, v_t, t)$ is
{$$|x_t - (x_0 + \int_0^{t} v(s) ds)|\leq \delta \wedge |v_t - (v_0 + \int_0^t g(1-\beta v(s)^2) ds))|\leq \delta$$}
and $\jump_{q_d \rightarrow q_u}^{\delta} (x, v, x', v')$ is {$$|x|\leq \delta \wedge |v' - \alpha v|\leq \delta \wedge |x'-x|\leq \delta.$$}
\end{example}
It is important to note that the notion of $\delta$-perturbations is a purely syntactic one, defined on the description of hybrid systems. Following Proposition~\ref{overap}, it can be easily seen that the syntactic perturbations correspond to semantic over-approximation of $H$ in the trajectory space.

\subsection{Complexity of Stability}
We now obtain complexity results for stability of hybrid systems. The main difference from the continuous systems is that the set of reachable states of a hybrid system requires a more complex encoding. However, we will see that they do not change the upper bound of the complexity, since the quantification structure does not change.  

First, we need to define a set of auxiliary formulas that will be important for ensuring that a particular mode is picked at a certain step.
\begin{definition}
Let $Q = \{q_1,...,q_m\}$ be a set of modes. For any $q\in Q$, and $i\in\mathbb{N}$, use  $b_{q}^i$ to represent a Boolean variable. We now define
$$\enforce_Q(q,i) = b^i_{q} \wedge \bigwedge_{p\in Q\setminus\{q\}}\neg b^{i}_{p}$$
$$\enforce_Q(q, q',i) = b^{i}_{q}\wedge \neg b^{i+1}_{q'} \wedge \bigwedge_{p\in Q\setminus\{q\}} \neg b^i_{p} \wedge \bigwedge_{p'\in Q\setminus\{q'\}} \neg b^{i+1}_{p'}$$
We omit the subscript $Q$ when the context is clear.\end{definition}
\begin{definition}[$k$-Step Reachable Set]
Suppose $H$ is invariant-free, and $U$ a subset of its state space represented by $\unsafe$. The $\lrf$-formula $\reach_{H,U}(k,M)$ is defined as:
\begin{eqnarray*}
& &\bigvee_{q\in Q} \Big(\init_{q}(\vec x_{0})\wedge \flow_{q}(\vec x_{0}, \vec x_{0}^t, t_0)\wedge \enforce(q,0)\wedge \forall^{[0,t_0]}t\forall^X\vec x\;(\flow_{q}(\vec x_{0}, \vec x, t)\rightarrow \inv_{q}(\vec x))\Big) \\
& &\wedge\; \bigwedge_{i=0}^{k-1}\bigg( \bigvee_{q, q'\in Q} \Big(\jump_{q\rightarrow q'}(\vec
x_{i}^t, \vec x_{i+1})\wedge \flow_{q'}(\vec x_{i+1}, \vec x_{i+1}^t, t_{i+1})\wedge \forall^{[0,t_{i+1}]}t\forall^X\vec x\;(\flow_{q'}(\vec x_{i+1}, \vec x,
t)\rightarrow \inv_{q'}(\vec x)) )\\
& &\hspace{9cm}\wedge \enforce(q,q',i)\wedge\enforce(q',i+1)\Big)\bigg)
\end{eqnarray*}
\end{definition}

\begin{proposition}[Hybrid Lyapunov Stability]
The origin is a stable equilibrium point if 
\begin{eqnarray*}
\forall^{[0,\infty)} \varepsilon\exists^{[0,\varepsilon]} \delta \forall^{[0,\infty)} t\forall x_0\forall x_t (||x_0||<\delta \wedge \reach(x_0,x_t,t) )\rightarrow ||x_t||<\varepsilon.
\end{eqnarray*}
\end{proposition}
\begin{proposition}[Asymptotic Stability]
The origin is an asymptotically stable equilibrium point if 
\begin{eqnarray*}
& &\forall^{[0,\infty)} \varepsilon\exists^{[0,\varepsilon]} \delta\forall^{[0,\infty)} t\forall x_0\forall x_t\Big((||x_0||<\delta \wedge \reach(x_0,x_t,t) )\rightarrow ||x_t||<\varepsilon\Big)\\
& &\wedge \exists^{[0,\infty)} \delta'  \forall^{[0,\infty)} t\forall x_0\forall x_t\Big( (||x_0||<\delta'\wedge \reach(x_0,x_t,t))\rightarrow \lim_{t\rightarrow \infty} ||x_t|| = 0\Big). 
\end{eqnarray*}
\end{proposition}
The definition is $\delta$-stability is the same as in the continuous case. 
\begin{definition}[$\delta$-Stability]\label{sl}
The (Lyapunov or asymptotic) $\delta$-stability problem of hybrid systems asks for one of the following answers:
\begin{itemize}
\item {\sf stable}: The system is stable. 
\item {\sf $\delta$-unstable}: Some $\delta$-perturbation of the $\lrf$-representation of stability is false. 
\end{itemize}
\end{definition}

\begin{theorem} Suppose all terms in the $\lrf$-representation of stability are in Type 2 complexity class $\mathsf{C}$. We have 
\begin{itemize}
\item The bounded Lyapunov $\delta$-stability problem of hybrid systems is in $\mathsf{(\Pi^P_3)^C}$. The asymptotic $\delta$-stability of hybrid systems is in complexity class $\mathsf{(\Sigma^P_4)^C}$. 
\item The unbounded $\delta$-stability problem of hybrid systems is in $\mathsf{\Pi_1^0}$ and asymptotic $\delta$-stability is in $\mathsf{\Sigma_2^0}$. 
\end{itemize}
\end{theorem}
From these results, it may seem that hybrid systems are not harder than continuous systems, in terms of the upper bounds on complexity. However, the discrete components of hybrid systems make it much easier to reach a high lower bound on the complexity. For instance, it is easy to show that the complexity results are tight in the following sense:
\begin{theorem}
Suppose the terms in describing the stability formulas are polynomial-time Type 2 computable. The bounded Lyapunov and asymptotic $\delta$-stability of hybrid systems are both $\mathsf{(\Pi^P_3)}$-complete. 
\end{theorem}
The reason is that logic formulas can be easily encoded as jumping conditions of hybrid systems. It is then straightforward to reduce complete problems in the complexity class to stability problems of hybrid systems. We omit the full proof here. 

\section{Conclusion and Future Work}
We defined a framework for measuring the ``practical complexity" of stability problems for a wide range of nonlinear continuous and hybrid systems. To do so, we describe stability properties of systems as first-order formulas over the real numbers, and reduce stability problems to the $\delta$-decision problems of these formulas. The framework allows us to obtain a precise characterization of the complexity of different notions of stability that has not been discovered previously. We prove that bounded version of the stability problems are generally decidable, and give precise measure of the upper bound of their complexity. The unbounded versions are generally undecidable, for which we give precise measures of their degrees of undecidability. 

We believe the results serve as a basis for developing computational methods towards nonlinear and hybrid control techniques. An immediate next step is to use these methods to study other problems such as controllability and observability of nonlinear systems. On the other hand, the logical descriptions of the problems can directly guide the development of practical decision procedures for the problems. 

\bibliographystyle{abbrv}
\bibliography{tau}

\end{document}